\documentclass[11pt]{extarticle}
\usepackage{latexsym}
\usepackage{theorem}
\usepackage{graphicx}
\usepackage{amsmath,color}
\usepackage{amsfonts}
\usepackage{natbib}
\usepackage{soul}
\usepackage{enumerate}

\usepackage{tikz}
\usepackage{caption}
\usepackage{subcaption}
\usepackage{comment}

\theorembodyfont{\rmfamily}
\newtheorem{theorem}{Theorem}

\newtheorem{lemma}[theorem]{Lemma}
\newtheorem{proposition}[theorem]{Proposition}
\newtheorem{corollary}[theorem]{Corollary}
\newtheorem{definition}[theorem]{Definition}

\theoremstyle{break}

\newtheorem{example}[theorem]{Example}

\usepackage{mathtools}

\newenvironment{proof}{\paragraph{Proof.}}{\hfill$\square$}

\title{Search for an Immobile Hider \\ on a Binary Tree with \\ Unreliable Locational Information}

\date{}
\author{Steve Alpern\thanks{Warwick Business School, University of Warwick, Coventry CV4 7AL, UK, steve.alpern@wbs.ac.uk} \and Thomas Lidbetter\thanks{Department of Management Science and Information Systems, Rutgers Business School, NJ 22902, USA, tlidbetter@business.rutgers.edu}}

\providecommand{\keywords}[1]{\textbf{\textbf{Keywords:}} #1}

\begin{document}
	
\maketitle

\begin{abstract}
Adversarial search of a network for an immobile Hider (or target) was
introduced and solved for rooted trees by \cite{Gal79}. In this zero-sum game, a Hider picks a point to hide on the tree and a Searcher picks a unit speed trajectory starting at the root. The
payoff (to the Hider) is the search time. In Gal's model (and many
subsequent investigations), the Searcher receives no additional information
after the Hider chooses his location. In reality, the Searcher will often receive such locational information. For
homeland security, mobile sensors on vehicles have been used to locate
radioactive material stashed in an urban environment. 
In a military setting, mobile sensors can detect chemical signatures from
land mines. In predator-prey search, the predator
often has specially attuned senses (hearing for wolves, vision for eagles,
smell for dogs, sonar for bats, pressure sensors for sharks) that may help
it locate the prey. How can such noisy locational information be used by the
Searcher to modify her route? We model such information as signals which
indicate which of two branches of a binary tree should be searched first,
where the signal has a known accuracy $p<1$. Our solution calculates which
branch (at every branch node) is \textit{favored}, meaning it should always
be searched first when the signal is in that direction. When the signal is
in the other direction, we calculate the probability the signal should be
followed. Compared to the optimal Hider strategy in the classic search game of Gal, the Hider's optimal distribution for this model is more skewed towards leaf nodes that are further from the root.
\end{abstract}


\newpage

\section{Introduction}

The game where a Searcher wishes to find a stationary adversarial Hider (or
target), starting from a designated location in a search region $S$, was
introduced in the classic book on Differential Games of \cite{Isaacs}.
Later, \cite{Gal79} solved the game where $S$ is a tree, considering the
Searcher starting point as the root. It is optimal for the Hider to locate at a leaf node, as
other locations are dominated. The optimal hiding distribution is the Equal
Branch Density (EBD) distribution, which locates in each branch at a branch node
with a probability proportional to the length of that branch. It is optimal for the Searcher,
when reaching a branch node for the first time, to choose to fully
search each branch equiprobably, and when returning to the branch node to
search the other branch. This description assumes a binary tree, but any tree can be transformed into a binary tree with the addition of some arcs of length zero. 
The Searcher thus searches in a depth-first manner and traces out a minimal length tour, or {\em Chinese
Postman Tour}. The value of the game is the total length $\mu $ of the tree.

A particular assumption made by Isaacs and Gal is that the Searcher receives
no additional information (on the Hider's location) during the search. In
reality, this assumption is often unwarranted, as the Hider might emit
signals (odor, radiation, sound) that the Searcher might be able to detect.
The aim of this paper is to see how such information changes the optimal
strategies of Searcher and Hider. If the signals given out by the Hider were
perfect, the solution is simple: the Hider locates at the furthest point
from the Searcher starting point and the Searcher proceeds directly to the
Hider's location. However in practice the Hider's location is a noisy signal
at best. Our model gives the Searcher a signal as to which branch at her
current branch node contains the Hider -- it is correct with known
probability $p$. When $p$ is $1/2$ (for a binary tree) the signal is useless
and the game reduces to that solved by Gal. When $p>1/2,$ we identify a
favored branch at each branch node: a signal that the Hider is in that
branch is always followed, while we determine the probability that the alternative signal should be followed. We also
determine the optimal hiding distribution, as a function of $p$. Examples
illustrating the solution are given in Section~\ref{sec:main}, before our formal
analysis. 

The motivation for giving the Searcher additional information in the form of
signals comes from real world examples and prior academic work which we put
into a game theoretic context for the first time here. In the field of
Homeland Security, \cite{Hochbaum14} analyzed the use of mobile sensors
in cities to locate cached nuclear material through the emitted radiation.
In a military setting, \cite{JA15} considering the detection of
landmines and improvised explosive devices. Sometimes specially trained dogs
can use their sense of smell to the same end \cite[see][]{Evans22}. The optimal
routing problem for mobile sensors has been explored by \cite{Paley16}. But
none of these investigations have been carried out in a game theoretic
context, where the target is hidden adversarially. 

In the field of predator search for prey (stationary or mobile), sensory
cues are important and explain in part the highly developed senses of many
predators. \cite{HM13} observe that ``Most motile organisms use
sensory cues when searching for resources, mates, or prey{\ldots}Yet,
classical models of species encounter rates assume that searchers move
independently of their targets.'' This failure also appears in some classical
hide-seek game models which we attempt to remedy here. Further work on
detection of targets during search is mentioned in our literature review
(Section~\ref{sec:lit}).

\section{Literature Review} \label{sec:lit}

We now give a short overview of work on network search games since \cite{Gal79} as well as some further examples of sensory detection of targets in a non game theoretic context. 

The pioneering work of \cite{Gal79} on tree search, described in our opening paragraph, has been extended in many ways. More general networks were analyzed by \cite{RP93} and \cite{Gal2001}. 
Computational methods for determining optimal strategies were given by \cite{AA}. 
The requirement for a designated Searcher starting point was removed by \cite{DG08}. \cite{AL13, AL14} studied search games
on windy networks and by expanding regions rather than paths. \cite{Alpern17} restricted the search paths on general networks to ``combinatorial'' ones consisting of sequences of edges. 
Other investigations considered search at nodes of a lattice \citep{Zoroa13} and costs for searching at nodes \citep{BK15}. The requirement to bring the target back to the root after capture (find-and-fetch) was considered by \cite{Alpern11}. \cite{Angel20} considered the {\em linear search problem} in the setting where the Searcher has a ``hint'' about the location of the target. For general discussions of search games, see \cite{AG03}, \cite{Garnaev} and \cite{Hohzaki16}.

There has been considerable work on the detection of targets by sensors (machines) or senses (animals). An abstract computer science approach is given in \cite{Patan12}. A search technique which combines vehicles and drones has been analyzed by \cite{GF19}. In the field of ecology, \cite{NON20} show how the ladybird beetle uses olfcatory signals in their walks on plants to find aphids. Females made more use of these signals than males. The location signals can sometimes come in multiple forms, as in \cite{Catania08}, where movement, shape and smell of prey can all be detected by water shrews.

\section{Illustration of Main Results} \label{sec:main}

Our main result is that for our game $G(Q,O,p)$, every branch node $j$ has a {\em favored branch}. With some probability $\beta=\beta(j)$ the Searcher searches the favored branch first, whatever the signal is. We call $\beta$ the {\em favoring bias}, and it is calculated according to Equation~({\ref{beta}) in Theorem~\ref{thm:value}.
With probability $1-\beta$ the Searcher follows the signal, whichever direction it points.
That is, she searches the branch the signal implies the Hider is in. 
When she arrives back at the branch node, she searches the other branch. 
In particular, the Searcher never searches (first) the unfavored branch when the signal indicates the favored branch. 
For simple two-arc trees, or penultimate nodes (whose two branches are both leaf arcs), the favored branch is the longer one. 
The hiding strategy is different that the one found by Gal for the no-signal case, but is also given by a recursive algorithm.

To illustrate the nature of the optimal Searcher strategy we describe the solution for the simplest tree which has more than two arcs. 

\begin{example} \label{ex1}
Consider the tree shown in Figure~\ref{fig:network}. The quantitative recursive calculations of the favoring biases will be given in Section~\ref{sec:example}, but for now we qualitatively describe the optimal Searcher strategy.  

\begin{figure}[h]
\center
\includegraphics[width=8cm]{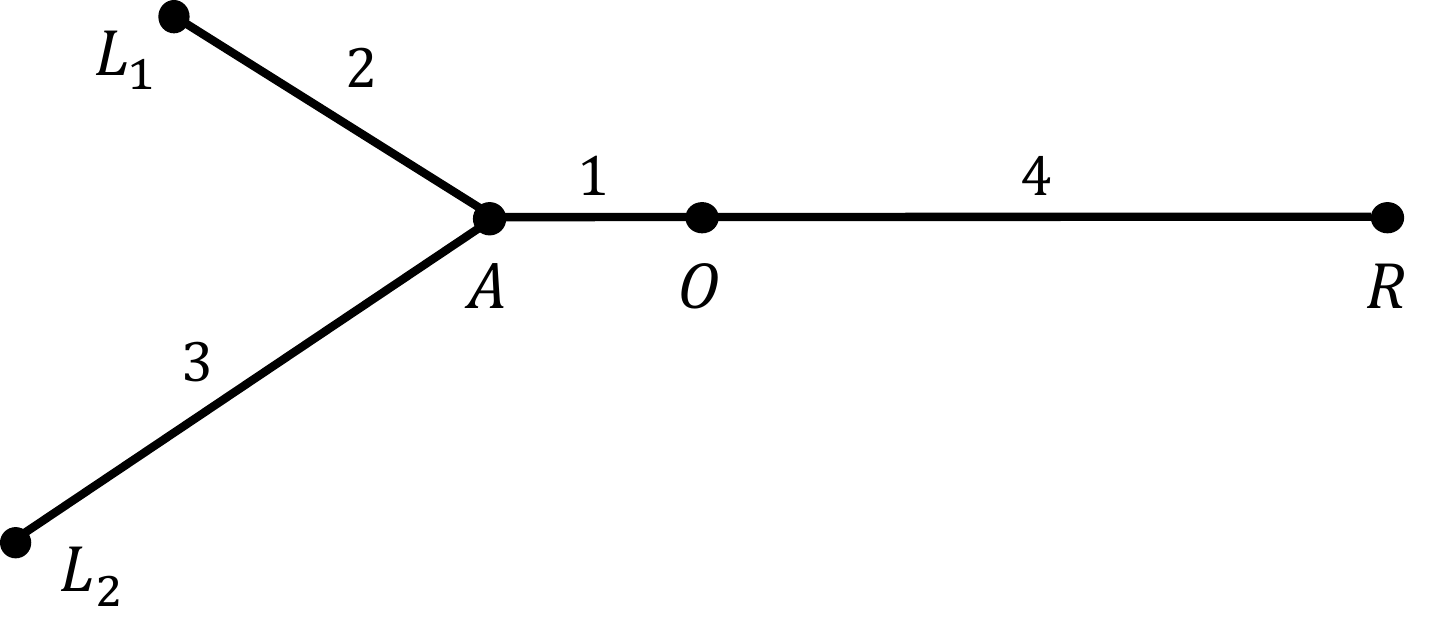}
\caption{A tree with root $O$.}
\label{fig:network}
\end{figure}

The tree of Figure~\ref{fig:network} has two branch points, $O$ (Searcher starting point) and $A$.
The Hider can choose one of the two leaf nodes on the left ($L_1$ and $L_2$) or the leaf node $R$ on the right, as all other points are clearly dominated by these.
Suppose that $p = 2/3$. 
The optimal Searcher strategy is given below in Figure~\ref{fig:network-sol}.

\begin{figure}[h]
\center
\includegraphics[width=8cm]{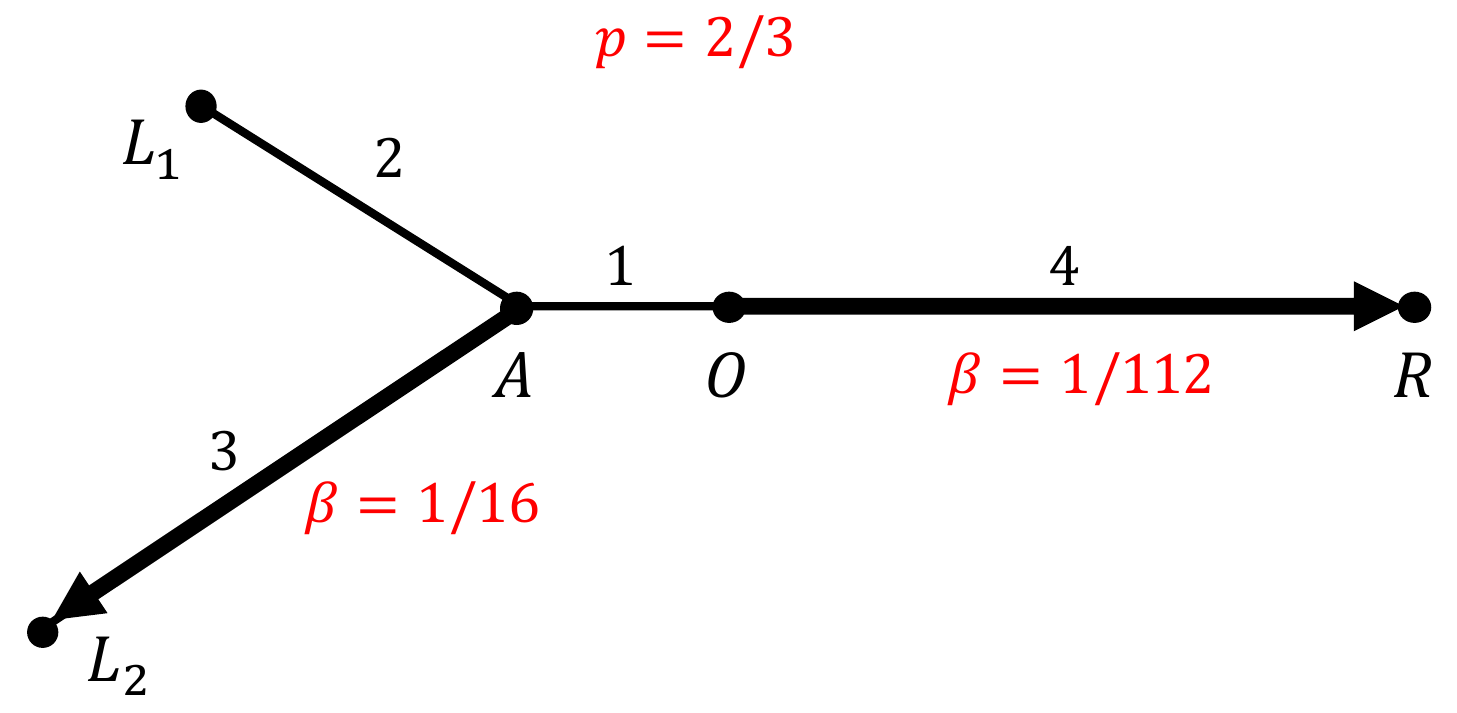}
\caption{A tree with root $O$, favored branches thickened.}
\label{fig:network-sol}
\end{figure}

We have to specify what the Searcher does, probabilistically, at each of her branch nodes $A$ and $O$. 
At node $A$, the favored branch is the longer arc $AL_2$ (indicated by the thickened line) and the favoring bias is given by $\beta =1/16$, which is calculated according to (\ref{beta}), as described later in Theorem~\ref{thm:value}.
So when the Searcher reaches $A$ she goes to $L_{2}$ before or without looking at the signal with probability $1/16$ and otherwise simply follows the signal. 
At $O$ the favored branch is the arc $OR$ and the favoring bias is $\beta =1/112$.
So at the start of the game, the Searcher goes to $R$ with probability $1/112$, otherwise she follows her signal.

\end{example}

\section{Formal Description of the Model}
\label{sec:model}

Here we give a more detailed description of the game  than in the Introduction. We define a game $G=G(Q,O,p)$, where  $Q$ is a tree, given as a metric space (in that arcs have lengths), $O$ is a point of $Q$ and $1/2<p\le1$. 
The Hider chooses a leaf node (other hiding strategies are dominated) and then the Searcher conducts some unit speed minimum length tour of the tree (that is, a depth-first search).
We assume that $Q$ is a binary tree: that is, each node has degree either three or one with the possible exception of the root $O$. (Any tree that is not binary can be converted to a binary tree by adding some arcs of length $0$.) We call every node of degree three a {\em branch node}, and if $O$ has degree two, we also call $O$ a branch node. At each branch node $j$, we call the set of points at least as far away from $O$ as $j$ the {\em subtree at $j$}. Removing $j$ from the subtree at $j$ partitions this subtree into two connected components we refer to as the {\em branches} at $j$.

When arriving at any branch node $j$, the Searcher receives a signal as to which branch, $Q_1$ or $Q_2$ the Hider belongs to. This signal is correct with the fixed and known probability $p$, and wrong with probability $q=1-p$. 
If the Hider lies in neither branch, any signal distribution may be used, as in this case the Searcher will return to node $j$ again after time equal to twice the length of the branches, regardless of her search method. As $p \rightarrow 1/2$, the signal becomes useless and the solution of the game reverts back to that of the game without any signals, as solved in \cite{Gal79}. 

The Searcher picks which branch to search with knowledge of the signal, but she does not have to
follow the signal (in fact it is almost never optimal to always follow the signal). 
The payoff is the time~$T$ for the Searcher to reach the Hider. A Searcher mixed strategy can
be given by specifying her choice of branch on her first
arrival to each branch node, dependent on the signal. It is a distribution
over four strategies, specifying which of two ways to go based on the binary signal:
choose $Q_{1}$ with any signal, choose $Q_{2}$ with any signal, follow the
signal, oppose the signal. The last unintuitive choice will indeed be shown
to be dominated.

The solution for the Searcher will have the following structure. At every branch node $j$ there is a favored branch $Q_1$ and a positive probability $\beta$ (the {\em favoring bias}) for it to be chosen before looking at the signal. 
With the remaining probability $1-\beta$ the search follows the signal. 
So in particular the Searcher will never choose the unfavored arc (branch) when the signal is for the favored one. The use of biased depth-first Searcher strategies (random choices at every branch node) of the Searcher was introduced in another context in \cite{Alpern10} and \cite{AL14}, but those distributions are not the same is in the present context.

The optimal Hider distribution over the leaf nodes can be found by a similar stochastic process in which the Hider starts at the root $O$ and at each branch node chooses a branch to enter according to a certain distribution. Of course, this is merely a mental calculation for the Hider, who is stationary in this game. 

Some qualitative findings are that for all penultimate nodes (where both branches are single arcs), the longer arc is favored, and that it is (almost) never optimal for the Searcher to simply follow her signal.

A final, fairly obvious, observation is that having signals cannot hinder the Searcher, because she could always ignore them. 
That is, the value of the game with signals cannot be larger than that of the classic search game without signals. 
Since the value of this classic search game on a tree is equal to its total length $\mu$ \citep{Gal79}, this must be an upper bound for the value of the search game with signals on a tree. 
For later purposes, we state the following.

\begin{lemma}[$V<\mu$] \label{lem:value}
The game $G(Q,O,p)$ has a value not larger than the total length $\mu$ of $Q$.
\end{lemma}

Later we can show that $V=\mu $ {\em implies} that the tree is a single arc (with $O$ at one end).

For the sake of completeness, we could establish Lemma~\ref{lem:value} directly by repeating Gal's idea of having the Searcher play an equiprobable mixture of a Chinese Postman Tour of the tree $Q$ and its time reversed tour. 
Such a mixed strategy reaches every point $H$ in $Q$ in expected time not more than the length $\mu $ of $Q$ because such a tour has length $2\mu$.

\section{Solution of the Game}
\label{sec:solution}

We first give some basic definitions that will enable us to state our main results on the value of the game $G(Q,O,p)$.

\begin{definition}
We define the {\em mean depth} $D$ of a rooted tree, with respect to a given probability measure $\lambda$ on its leaf nodes, as the mean distance from the root to the set $\mathcal{L}$ of leaf nodes, weighted with respect to $\lambda$. More precisely, we define
\[
D = D_Q = D(Q,O,\lambda) = \sum_{v \in \mathcal{L}}  \lambda(v) ~ d(O,v),
\]
\end{definition}

In this paper we will take $\lambda$ to be the optimal Hider strategy $\bar{\lambda}$. We will show that $\bar{\lambda}$ is optimal in Theorem~\ref{thm:value}. Clearly, $D \le D_{\max} \equiv \max_{v \in \mathcal{L}} d(O,v)$, where $D_{\max}$ is usually called the {\em depth} of the tree.

\begin{definition}[Optimal Hider distribution $\bar{\lambda}$.] The optimal Hider distribution $\bar{\lambda}=\bar{\lambda}_Q$ on a binary tree $Q$ with root $O$ is concentrated on the leaf nodes of $Q$ and will be defined recursively. 
If $Q$ has only one arc $OA$, then $\bar{\lambda}_Q(A)=1$.

Suppose now that $\bar{\lambda}_Q$ has been defined for all trees with at most $m$ arcs. Let $Q$ have $m+1$ arcs. If $O$ has degree $1$ with an arc leading to a branch $Q'$, then $Q$ and $Q'$ have the same leaf nodes, and the Hider distribution on these is the same.

Now suppose that $O$ is a branch node with branches $Q_1$ and $Q_2$ rooted at $O$ whose optimal hiding distributions are $\bar{\lambda}^1$ and $\bar{\lambda}^2$, respectively. Let $\mu_1$ and $\mu_2$ denote their respective total lengths. Without loss of generality, suppose that $Q_1$ has greater mean depth than $Q_2$, that is $D(Q_1,\bar{\lambda}^1) \ge D(Q_2,\bar{\lambda}^2)$.

In this case, we will call $Q_1$ the {\em favored branch}. We now define the optimal Hider distribution on $Q$ (on its leaf nodes) by the formula
\begin{align}
\bar{\lambda}(v) = 
\begin{cases}
\frac{p \mu_1}{p \mu_1 + q \mu_2} \bar{\lambda}^1(v) \text{ if $v$ is a leaf node of $Q_1$ and} \\
\frac{q \mu_2}{p \mu_1 + q \mu_2} \bar{\lambda}^2(v) \text{ if $v$ is a leaf node of $Q_2$.}
\end{cases}
\label{eq:lambda-bar}
\end{align}
\end{definition}

We illustration the optimal Hider strategy with the following simple example

\begin{example}
Suppose $Q$ is a tree with two arcs, $OA$ of length $3$ and $OB$ of length $5$. With no signals ($p=1/2$), the EBD distribution of \cite{Gal79} says that the optimal probabilities of hiding at $A$ and $B$ are proportional to their lengths: that is, $3/8$ and $5/8$, respectively. Since $D_{OB}=5>3=D_{OA}$, clearly $OB$ is the favored branch, so on $Q$ we have
\[
\bar{\lambda}(B) = \frac{p \cdot 5}{p \cdot 5 + q \cdot 3}(1) = \frac{5p}{2p+3} > \frac{5}{8},
\]
where the inequality follows from our assumption $1/2 < p <1$. Thus we see that when there are signals the weights on leaf nodes are no longer proportional, but skewed further towards the longer branches. So for this tree we have that the mean depth is given by
\[
D=\bar{\lambda}(A) \cdot 3 + \bar{\lambda}(B) \cdot 5 =  \frac{3 - 3p}{2p+3} (3) + \frac{5p}{2p+3}(5) = \frac{16p+9}{2p+3}.
\]

We observe that as $p$ goes to $1$ and $q$ to $0$, the distribution of $\bar{\lambda}$ becomes concentrated on the leaf node at greatest distance from $O$, and $D$ converges to that distance. As $p$ goes to $1/2$, the distribution of $\bar{\lambda}$ converges to the EBD distribution.

\end{example}

Before stating our main theorem, we make an observation.

\begin{proposition} \label{prop:Delta}
Suppose $Q$ is a tree with root $O$ and first suppose that $O$ has degree $1$. Let $A$ be the neighbor of $O$ and let $Q'$ be the subtree rooted at $A$. Let $\ell$ denote the length of the arc $OA$. Then 
\[
D_Q=\ell + D_Q'.
\]
Now suppose $O$ has degree 2 and let $Q_1$ and $Q_2$ be the branches at $O$ with lengths $\mu_1$ and $\mu_2$, respectively and suppose $Q_1$ is the favored branch (${D_{Q_1} \ge D_{Q_2}}$). Then
\[
D_Q = \frac{p \mu_1 D_{Q_1} + q\mu_2 D_{Q_2}}{p\mu_1+q\mu_2}.
\]
\end{proposition}
\begin{proof}
For the first part of the proposition, we calculate
\begin{align*}
D_Q &=  \sum_{v \in \mathcal{L}}  \bar{\lambda}(v) ~ d(O,v) \\
&=  \sum_{v \in \mathcal{L}}  \bar{\lambda}(v) (\ell + d(A,v)) \\
&= \ell + \sum_{v \in \mathcal{L}}  \bar{\lambda}(v) ~ d(A,v) \\
&= \ell+ D_Q',
\end{align*}
where the penultimate equality follows from the fact that $\bar{\lambda}$ is a probability distribution on $\mathcal{L}$ and the final equality follows from the fact that the leaf nodes of $Q$ and $Q'$ are the same.

For the second part of the proposition, let $\mathcal{L}_1$ and $\mathcal{L}_2$ be the leaf nodes of $Q_1$ and $Q_2$, respectively. Also, let $\bar{\lambda}^1$ and $\bar{\lambda}^2$ be the optimal Hider strategies on $Q_1$ and $Q_2$, respectively.  Then
\begin{align*}
D_Q &=  \sum_{v \in \mathcal{L}_1}  \bar{\lambda}(v) ~ d(O,v) +  \sum_{v \in \mathcal{L}_2}  \bar{\lambda}(v) ~ d(O,v) \\
&= \left( \frac{p\mu_1}{p\mu_1 + q\mu_2} \right) \sum_{v \in \mathcal{L}_1}  \bar{\lambda}^1(v) ~ d(O,v) +   \left(\frac{q\mu_2}{p\mu_1 + q\mu_2} \right) \sum_{v \in \mathcal{L}_2}  \bar{\lambda}^2(v) ~ d(O,v) \\ 
&= \frac{p \mu_1 D_{Q_1} + q\mu_2 D_{Q_2}}{p\mu_1+q\mu_2},
\end{align*}
where the penultimate equality follows from the definition of $\bar{\lambda}$ and the final equality follows from the definition of $D_{Q_1}$ and $D_{Q_2}$.
\end{proof}

We can now state and prove our main theorem, which includes an expression for the value of the game. We describe the optimal strategy for the Searcher by giving the favoring bias $\beta$ of searching the favored branch first (without needing to observe the signal) when at a branch node.
\begin{theorem} \label{thm:value}
Let $(Q,O)$ be a rooted tree with length $\mu = \mu(Q)$. 
\begin{enumerate}[(i)]
\item The value $V=V(Q)$ of the  game $G(Q,O,p)$ is
\begin{equation}
V(Q)=2q  \mu + (p-q)D_Q.
\label{V}
\end{equation}
\item The hiding distribution $\bar{\lambda}$ is optimal for the Hider. 
\item When at a branch node with branches $Q_1$ and $Q_2$ of lengths $\mu_1$ and $\mu_2$ it is optimal for the Searcher to search the favored branch first with probability 
\begin{equation}
\beta=\frac{(p-q)(D_{Q_1}-D_{Q_2})}{2(p \mu_1 + q \mu_2)},
\label{beta}
\end{equation}
 where, without loss of generality, $Q_1$ is the favored branch. With the complementary probability $1-\beta$ the Searcher follows the signal, whichever direction it points.
\end{enumerate}
\end{theorem}

\begin{proof} 
The proof is by induction on the number of arcs of $Q$. If $Q$ has only one arc of length $\mu$, then $D(Q)=\mu$ and the right-hand side of~(\ref{V}) simplifies to $\mu$, which is trivially equal to the value of the game. Each player only has one strategy.

Now suppose the theorem is true for all trees with up to $m$ arcs, and let $Q$ have $m+1$ arcs. First suppose $O$ has degree one, and let $A$ be its neighbor. Let $Q'$ be the subtree rooted at $A$ and denote its length by $\mu'$. Let $\ell$ be the length of the arc $OA$. Then it is easy to see that $V(Q)=\ell+V(Q')$, so by the induction hypothesis,
\[
V(Q)=\ell+ 2q  \mu' + (p-q)D_{Q'} = \ell + 2q(\mu-\ell) +(p-q)D_{Q'}.
\]
By Proposition~\ref{prop:Delta}, we have $D_{Q'}=D_Q-\ell$, so
\[
V(Q) = \ell + 2q(\mu-\ell) + (p-q)(D_Q - \ell) = 2q \mu + (p-q)D_Q.
\]
Thus, part (i) is proven, and parts (ii) and (iii) are trivially true.

Finally, suppose that $O$ is a branch node with branches $Q_1$ and $Q_2$, where $Q_1$ is the favored branch. Let $\mu_1$ and $\mu_2$ be the lengths of $Q_1$ and $Q_2$, respectively.
Since we assume the Searcher chooses a depth-first search, the Hider has two strategy classes. He can play the optimal mixed strategy on the leaves of $Q_1$ or the optimal mixed strategy on the leaves of $Q_2$. 
For the Searcher, she has two possible signals ($1$ or $2$) and two searches (search $Q_{1}$ first or search $Q_{2}$ first). 
This gives her four strategies $\left[ X,Y\right] $ where $X$ is the index of the subtree she searches first with signal $1$ and $Y$ is the index of the subtree she searches first with signal $2$. 
We call the strategy $\left[ 1,2\right]$ ``follow'' (go with the signal) and the strategy $\left[ 2,1\right]$ ``opposite''. 
It is easy to see that the resulting matrix game is given by Table~\ref{tab:matrix}, where $V_1$ is the value of the game played on $Q_1$ and $V_2$ is the value of the game played on $Q_2$.
\begin{table}[h!]
\caption{\label{tab:matrix} Payoff matrix.} 
\center
\begin{tabular}{l|l|l|l|l}
Hider \textbackslash Searcher & $[1,1]$ & $[2,2] $ & $[1,2] =$ ``follow'' & $[2,1]=$ ``opposite''
  \\ \hline
1 & $V_1$ & $2\mu_2+V_1$ & $q\cdot 2\mu_2 +V_1 $ & $p \cdot 2\mu_2 + V_1 $   \\ 
2 & $2\mu_1+V_2$ & $V_2$ & $q \cdot 2\mu_1 +V_2 $ & $p \cdot 2\mu_1 + V_2$  
\end{tabular}
\end{table}

We explain only one entry, that of row 2 and column 3 (``follow''). The Hider is in $Q_{2}$ so with probability $p$ the signal is $2$, in which case the Searcher goes to $Q_{2}$ first and finds the Hider in expected time $V_2$. 
With probability $q$ the Searcher wastes time $2 \mu_1$ searching $Q_{1}$, covering every arc twice, before finding the Hider in
additional expected time $V_2$. 
We now show that ``follow'' dominates ``opposite'' for any $p>1/2$ by the simply taking the last column of the matrix from the third column.
\[
\left( 
\begin{array}{c}
p\cdot 2\mu_2 +V_1   \\ 
p \cdot 2\mu_1 +V_2 
\end{array}
\right) -\left( 
\begin{array}{c}
q \cdot 2\mu_2+ V_1 \\ 
q \cdot 2\mu_1 + V_2
\end{array}
\right)  \\
=\left( 
\begin{array}{c}
2\mu_2(2p-1) \\ 
2\mu_1(2p-1) 
\end{array}
\right).
\]
Since $p>1/2$, it follows that $2p-1>0$. So ``opposite'' is dominated by ``follow''.

If the Hider chooses row 1 (hide optimally in $Q_1)$ with
probability $x$, the payoffs (expected capture times) $T_{\left[ 1,1%
\right] },T_{\left[ 2,2\right] }$ and $T_{\left[ 1,2\right] }$
corresponding to the Searcher's three undominated columns are given as
linear functions of $x$, as 
\begin{align*}
T_{[1,1]}(x)  &=  xV_1 +(1-x)(2\mu_1+V_2) \\
T_{[2,2]}(x)  &= x(2\mu_2+V_1) + (1-x)V_2 \\
T_{[1,2]}(x) &= x (q \cdot2\mu_2 +V_1) +(1-x)(q \cdot 2\mu_1+V_2).
\end{align*}
The payoff functions of the two no-signal strategies $T_{[1,1]}$ and $T_{[2,2]}$ intersect at the value
\[
\bar{x}=\frac{\mu_1}{\mu_1+\mu_2},\text{ with }T_{[1,1]}\left( \bar{x}\right) =T_{[2,2]}\left( 
\bar{x}\right) =\frac{\mu_1 V_1+\mu_2V_2 + 2\mu_1 \mu_2}{\mu_1+\mu_2}.
\]
Note that $T_{[1,1]}$ is decreasing, because its slope is%
\[
V_1-V_2-2\mu_1 \leq \mu_1-V_2-2\mu_1 = -(V_2 +\mu_1) <0,
\]
where the first inequality follows from Lemma~\ref{lem:value}. 

Similarly, $T_{[2,2]}$ is increasing because its slope is%
\[
V_1-V_2+2\mu_2 \ge V_1-\mu_2+2\mu_2 = V_1+\mu_2 >0,
\]
again applying Lemma~\ref{lem:value}. 

Finally, $T_{[1,2]}$ is decreasing, since its slope is
\[
q \cdot 2\mu_2 +V_1 - q \cdot 2\mu_1  -V_2 = (p-q)(D_{Q_1}-D_{Q_2}) > 0 
\]
where the  equality follows from the induction hypothesis and the inequality follows from the fact that $Q_1$ is the favored branch and $p > q$.

At the point $\bar{x},$ the line $T_{\left[ 1,2\right] }\left( x\right) $ lies below the intersection of the lines $T_{[1,1]}(x)$ and $T_{[2,2]}(x)$.
This can be seen by calculating
\[
T_{[1,2] }( \bar{x})  = \frac{\mu_1V_1+\mu_2 V_2+4q\mu_1 \mu_2}{\mu_1+\mu_2}\\
\]
and observing that
\begin{align*}
T_{[2,2]}( \bar{x}) -T_{[ 1,2] }( \bar{x})  &=\frac{\mu_1V_1+\mu_2 V_2+2\mu_1 \mu_2}{\mu_1+\mu_1}-\frac{\mu_1 V_1+\mu_2 V_2+4q \mu_1 \mu_2}{\mu_1+\mu_2}\\
&=\frac{2\mu_1 \mu_2( 1-2q) }{\mu_1+\mu_2}>0,
\end{align*}
because $q<1/2$. The full picture is sketched out in Figure~\ref{fig:graph}.
\begin{figure}[h]
\center
\includegraphics[width=8cm]{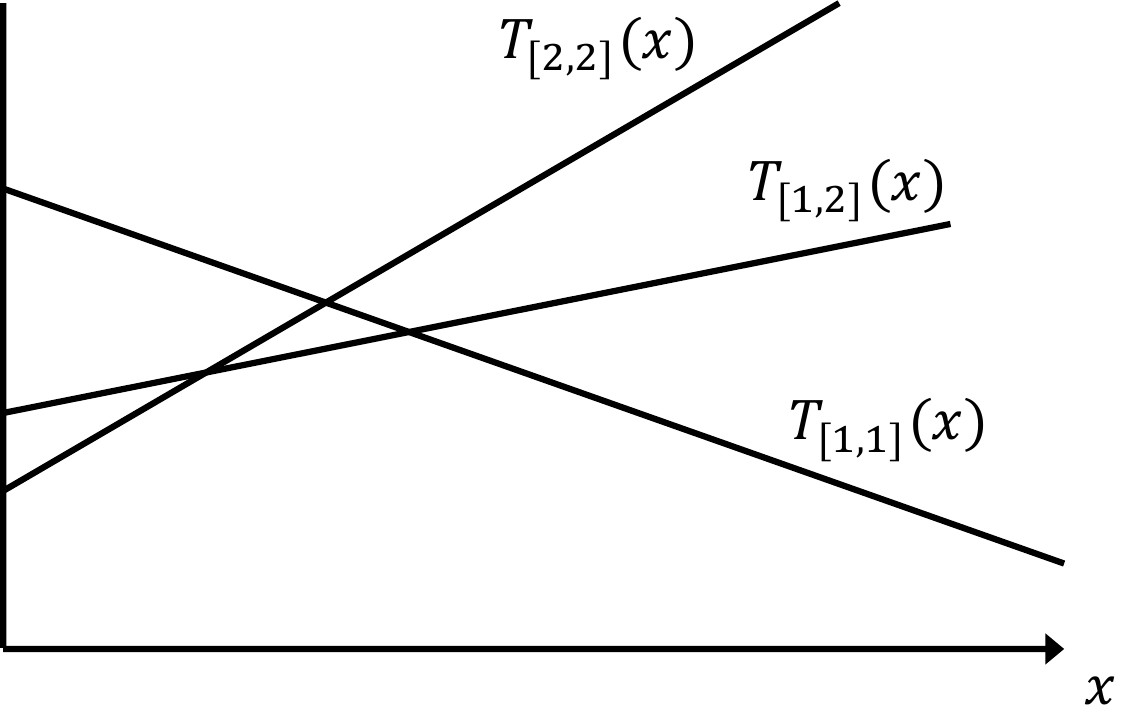}
\caption{A sketch of the functions $T_{[1,1]}(x)$, $T_{[2,2]}(x)$ and $T_{[1,2]}(x)$.}
\label{fig:graph}
\end{figure}

It is clear that the minimum of $T_{[1,1]}(x)$, $T_{[2,2]}(x)$ and $T_{[1,2]}(x)$ has a unique maximum at the intersection of $T_{[1,1]}(x)$ and $T_{[1,2]}(x)$, which occurs at the point
\[
x^* = \frac{p\mu_1}{p\mu_1+q\mu_2} = \bar{\lambda}(Q_1).
\]
This is therefore the optimal choice of $x$ for the Hider, and we have proven part (iii) of the theorem. 

Applying the induction hypothesis, the value of the game is
\begin{align*}
V = T_{[1,1]}(x^*)  &= \frac{p\mu_1V_1}{p\mu_1+q\mu_2} +\frac{q\mu_2(2\mu_1+V_2)}{p\mu_1+q\mu_2} \\
&= \frac{p\mu_1(2q\mu_1 + (p-q)D_{Q_1})}{p\mu_1+q\mu_2} +\frac{q\mu_2(2\mu_1+2q \mu_2 + (p-q)D_{Q_2})}{p\mu_1+q\mu_2} \\
& = 2q \mu + (p-q)D_Q.
\end{align*}
by Proposition~\ref{prop:Delta}. We have proven part (i) of the theorem.

For part (ii), the Searcher's optimal strategy must mix between her two best responses to the Hider strategy calculated above: that is between strategies 1 and [1,2] (``follow''). Using the principle of indifference we calculate the probability $\beta$ of searching the favored branch $Q_1$ first which makes the Hider indifferent between hiding in the two subtrees $Q_1$ and $Q_2$.
\begin{align*}
\beta V_1+(1-\beta) (q\cdot 2 \mu_2 +V_1) & = \beta(2\mu_1+V_2) + (1-\beta)(q \cdot 2\mu_1 +V_2), \text{ so} \\
\beta( 2p\mu_1  + 2q\mu_2 ) &=V_1 - V_2  -2q(\mu_1-\mu_2)\text{ and}\\
\beta &= \frac{ (p-q)(D_{Q_1}-D_{Q_2})}{2(p\mu_1  + q\mu_2) },
\end{align*}
by the induction hypothesis. Thus, we have established (\ref{beta}).
\end{proof}

Note that as the signal becomes more certain, so that $p \rightarrow 1$ and $q \rightarrow 0$, the value (\ref{V}) goes to $D_{Q}$, which goes to the distance of the furthest leaf node from $O$. As previously remarked, the distribution of $\bar{\lambda}$ becomes concentrated on the leaf node at furthest distance from $O$ (assuming the generic case where there is a unique such node). The favored branch will always be the one containing this leaf node, and the signal will always be accurate, so the Searcher will always take this branch.

As the signal becomes more uncertain, so that $p ,q \rightarrow 1/2$, the value (\ref{V}) goes to $\mu$, the hiding distribution $\bar{\lambda}$ goes to the EBD distribution, and the probability $\beta$ goes to 0, so that the Searcher follows the signal (which is determined by the toss of a fair coin) with probability 1.

We also note that if the value of the game is $\mu$, then $\mu=2q\mu+(p-q)D_Q$, so that
\begin{align}
(p-q) \mu = (p-q) D_Q \label{eq:value-mu}
\end{align}
As we remarked earlier, $D_Q \le D_{\max}$ and also $D_{\max} \le \mu$. Since $p > q$, we must have $D_{Q}=D_{\max}=\mu$ for (\ref{eq:value-mu}) to hold, so $Q$ must be a single arc with $O$ at one end.

\begin{corollary} \label{cor:pen-node}
Suppose a penultimate node has branches (leaf arcs) of length $\ell,s$, with $s<\ell$.
Then the favoring bias $\beta$ to the favored long arc is given by 
\begin{align}
\beta &=\frac{(p-q)(\ell-s)}{2(p\ell+qs)},
\label{fpen}
\end{align}
the optimal probability of hiding in the long arc is given by
\begin{align}
x^*&=\frac{p \ell }{p \ell+qs }  \label{xf pen}
\end{align}
and the value of the game by 
\begin{align}
V&= 2q(\ell+s) + \frac{(p-q) (p \ell^2 + qs^2)}{p\ell +qs}. \label{Vpen}
\end{align}
\end{corollary}

\begin{proof}
For leaf arcs the game values are simply the arc lengths and the favored arc $Q_1$ is the arc of length $\ell$, so we have $V(Q_1)=\ell$ and $V(Q_2)=s$. 
So the formula (\ref{beta}) for the favoring bias $\beta$ becomes (\ref{fpen}). 
Observe that as $p,q \rightarrow 1/2$, we have $\beta \rightarrow 0$ so that the Searcher just follows the equiprobable signal, randomly choosing which branch to search first. 
The optimal probability of hiding in the long (favorite) branch, given by the Weighted Branch Density distribution, becomes (\ref{xf pen}), which in the no-signal $\left(p,q \rightarrow 1/2\right) $ model reduces to $\ell/( \ell+s) $ (hiding in each arc with a probability proportional to its length). 
Finally, the formula (\ref{V}) reduces to (\ref{Vpen}) which simplifies to the total length $\ell+s$ in the no signal case of $p,q \rightarrow 1/2$.
\end{proof}

We end this section by showing that the value of the game is non-increasing in $p$. This is intuitively obvious, since the higher $p$ is, the more reliable the signal is, so one would expect the search time to go down.

\begin{proposition} \label{prop:non-inc}
The value of the game is non-increasing in $p$.
\end{proposition}
\begin{proof}
Suppose $1 \ge p > p' > 1/2$, and let $V$ and $V'$ be the respective values of the games $G(Q,O,p)$ and $G(Q,O,p')$.  We will show that we can generate a signal which is correct with probability $p'$ by using a signal that is correct with probability $p$. In this way, we can use the Searcher strategy for $p'$ in the game $G(Q,O,p)$, thereby ensuring that we can find the Hider in expected time $V'$ in this game.
 
When the Searcher is at a branch node in game $G(Q,O,p)$ and receives a signal, we create a new signal which is equal to the received signal with probability $x$ and chosen uniformly at random with probability $1-x$.  Then the probability this signal is correct is $xp+(1-x)/2$.  Choose $x$ so that $xp+(1-x)/2 = p'$, and we obtain a signal that is correct with probability $p'$.  The precise value of $x$ is $(p'-1/2)/(p-1/2)$.

Thus the Searcher can ensure an expected search time of at most $V'$ in $G(Q,O,p)$ so that $V \le V'$.
\end{proof}

Note that if $Q$ is a single arc with the root at one of its ends, the value of the game is $\mu$ for any value of $p$, so we cannot say that the game is strictly decreasing in $p$ in general.

\section{Application of the Recursion Techniques}
\label{sec:example}

We now show how the recursion techniques of the previous section can be used to solve the game on the tree of Example~\ref{ex1} (with $p=2/3$), obtaining the optimal Searcher strategies of Figure~\ref{fig:network-sol}.
The rooted tree of Figure~\ref{fig:network} has two branch nodes, $A$ and $O$. 
The recursion works backwards from penultimate nodes to the root, so we start with the subtree at $A$: the arcs $A L_{1}$ and $AL_{2}$.
Since $A$ is a penultimate node, the two branches have values equal to their
lengths $\ell=3$ and $s=2$. Equation~(\ref{fpen}) of Corollary~\ref{cor:pen-node} says that the long arc of length $\ell=3$ is favored and the favoring bias is
\[
\beta=\frac{(p-q)(\ell-s)}{2(p\ell+qs)} = \frac{ (2/3-1/3)(3-2)}{2(2/3(3)+1/3(2))}=1/16,
\]
as indicated in Figure 2. 
Thus if the Searcher arrives at branch node $A$, she goes to leaf node $L_{2}$ immediately with probability $1/16$ without looking at her signal. 
Otherwise she follows the signal. 

The optimal Hider distribution, $\bar{\lambda}$ for the subtree $Q_1$ with arcs $AL_1$ and $AL_2$ is to choose $L_2$ with probability proportional to $p \cdot 3 = 2$ and $L_1$ with probability proportional to $q \cdot 2 = 2/3$, by Equation~(\ref{eq:lambda-bar}). That is, $L_2$ is chosen with probability $3/4$ and $L_1$ is chosen with probability $1/4$. Hence, 
\[
D_{Q_1}= 3/4 \cdot 3 + 1/4 \cdot 2 = 11/4.
\]
Now consider the original game starting at node $O$.
The left branch $Q_2$ at $O$ has 
\[
D_{Q_2}=1+D_{Q_1}=15/4.
\]
The right branch $Q_3$, which is simply the arc $OR$ has $D_{Q_3}=4 > D_{Q_2}$, so it is the favored branch, as indicated in Figure~\ref{fig:network-sol}.
The favoring bias $\beta(O)$, given by (\ref{beta}), is 
\[
\beta(O)=\frac{(p-q)(D_{Q_3}-D_{Q_2})}{2(p \mu(Q_3)+q \mu(Q_2))} =  \frac{(2/3-1/3)(4-15/4)}{2(2/3 \cdot 4+1/3 \cdot 6)} = 1/112.
\]
It is optimal for the Hider to choose $Q_3$ and $Q_2$ with probability proportional to $p \cdot 4=8/3$ and $q \cdot 6 = 2$, respectively. That is, $Q_3$ is chosen with probability $4/7$ and $Q_2$ is chosen with probability $3/7$. Putting this together with the optimal Hider strategy on $Q_1$, we get the optimal hiding probabilities
\[
\bar{\lambda}(R)=4/7, \quad \bar{\lambda}(L_1)= 3/7 \cdot 1/4 = 3/28, \quad \bar{\lambda}(L_2)=3/7 \cdot 3/4 = 9/28.
\]
Also, we have
\[
D_Q=4/7 \cdot 4 + 3/28 \cdot 3 + 9/28 \cdot 4= 109/28.
\]
Hence, by~(\ref{V}), the value of the game is
\[
V = 2q \mu + (p-q)D_Q = 2/3 \cdot 10 + (2/3-1/3)\cdot109/28 = 223/28 \approx 7.96.
\]

\section{Trees of Constant Depth}

We say that a tree has {\em constant depth} if there is an $r$ such that for all leaf nodes $v$, we have $d(O,v)=r$. If $Q$ has more than one arc, the root $O$ will be the unique center of $Q$ and $r$ will be the radius. For trees of constant depth $r$, the mean depth $D(Q,O,\lambda)$ is obviously equal to~$r$. The next corollary follows immediately from Equation~(\ref{V}).
\begin{corollary}
If $Q$ is a rooted tree with constant depth $r$, the value $V$ of the game $G(Q,O,p)$ is given by
\[
V= 2q \mu + (p-q)r.
\]
\end{corollary}
So as $p$ varies from the no-signal value $p=1/2$ to the full information value $p=1$, the value of the game varies from $\mu$ to $r$.

To consider the effect of signals, consider a family of {\em perfect} binary trees $B_n$. The tree $B_1$ consists of two arcs meeting at the root $O$. Assuming $B_n$ has been defined, $B_{n+1}$ consists of two  arcs meeting at the root $O$, along with a copy of $B_n$ attached (at its root) to the other end of each of these arcs. The tree $B_3$ is depicted in Figure~\ref{fig:B3}.

\begin{figure}[h]
\center
\includegraphics[width=8cm]{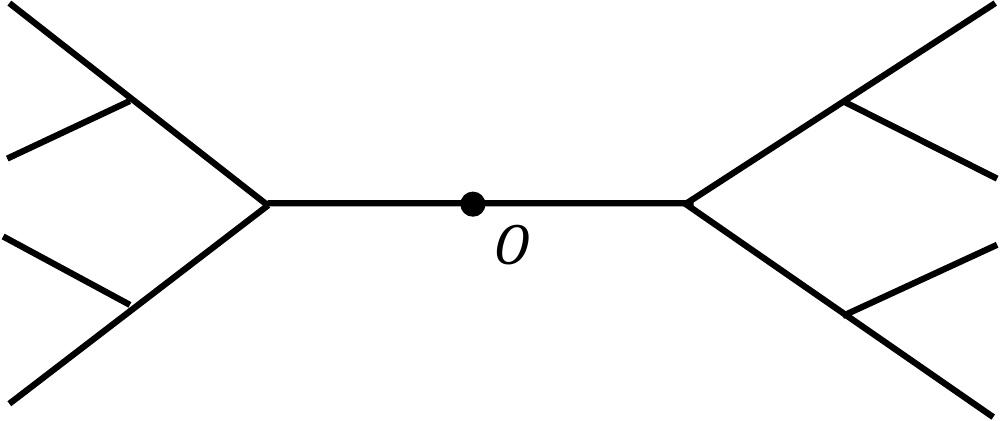}
\caption{The tree $B_3$.}
\label{fig:B3}
\end{figure}

The tree $B_n$ has $2^n$ leaf nodes and $2^{n+1}-2$ arcs. Suppose we take each arc to have length $\ell=1/(2^{n+1}-2)$, so that the total length $\mu$ of $B_n$ is $1$. Without signals, they are equally easy to search, and the value of the game is $\mu=1$. Now fix any $p>1/2$. The tree $B_n$ has constant depth $n\ell=n/(2^{n+1}-2)$, which is decreasing in $n$ and converges to $0$ as $n\rightarrow \infty$. By Theorem~\ref{thm:value}, we have
\begin{align*}
V & = 2q\mu  + (p-q)D \\
&= 2q + (p-q)n\ell \\
&\rightarrow 2q,
\end{align*}
as $n \rightarrow \infty$. So the value is decreasing in $n$ and converges to $2q$. For $p=2/3$, we have $V=2/3+(1/3)n/(2^{n+1}-2)$, converging to $2/3$, as plotted in Figure~\ref{fig:graph} for $n \ge 2$.

\begin{figure}[h]
\center
\includegraphics[width=8cm]{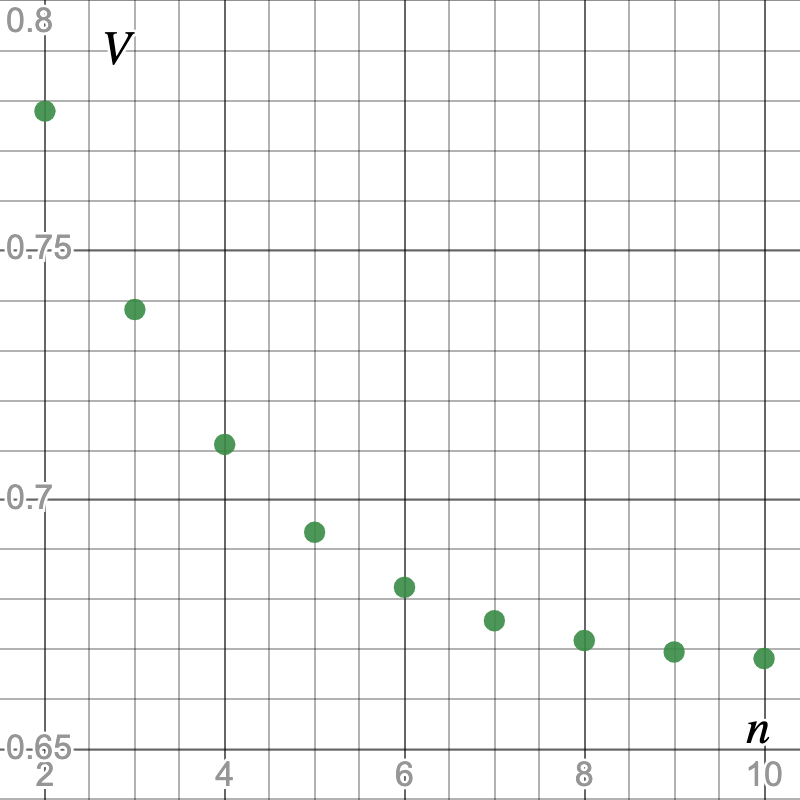}
\caption{Value of the game on $B_n$ for $p = 2/3$}
\label{fig:graph}
\end{figure}

\section{Conclusion} 

This paper has addressed the question of how to optimally search for an adversarially hidden target on a tree network in the presence of unreliable signals. We have found optimal solutions for both the Searcher and the Hider that can be calculated recursively, and a closed form expression for the value of the game. Future work might consider a variation of the game we consider here in which the time to traverse an arc depends on the direction of travel, as in the variable speed networks studied in \cite{AL14}.

\section*{Acknowledgements}
This material is based upon work supported by the Air Force Office of Scientific Research under award number FA9550-23-1-0556. 

The authors would also like to acknowledge the Lorentz Center at Leiden University, since some of the results were obtained at a Workshop on Search Games organized by the Lorentz Center.





\end{document}